\documentclass[12pt, draftcls, onecolumn, twoside]{IEEEtran}
\usepackage{hyperref}
\usepackage{amsmath}
\usepackage{amssymb}
\usepackage{amsfonts}
\usepackage{epsfig}
\usepackage{subfigure}
\usepackage{calc}
\usepackage{color}
\usepackage[all]{xy}
\usepackage{bm}
\usepackage{algorithmicx}
\usepackage{algpseudocode}
\usepackage{algorithm}
\usepackage{enumerate}
\usepackage{pdfsync}
\usepackage{framed}
\usepackage{caption}
\usepackage{amsthm}
\usepackage{blkarray}
\usepackage{float}
\newtheorem{defn}{Definition}

\newtheorem{thm}{Theorem}[section]
\newtheorem{cor}[thm]{Corollary}
\newtheorem{prop}{Proposition}
\newtheorem{claim}{Claim}
\newtheorem{lem}[thm]{Lemma}
\newtheorem{conj}[thm]{Conjecture}
\newtheorem{constr}[thm]{Construction}
\newtheorem{note}{Remark}
\newtheorem{example}{Example}
\newcommand{\bit}{\begin{itemize}}
\newcommand{\eit}{\end{itemize}}
\newcommand{\bcor}{\begin{cor}}
\newcommand{\ecor}{\end{cor}}
\newcommand{\beq}{\begin{equation}}
\newcommand{\eeq}{\end{equation}}
\newcommand{\beqn}{\begin{equation*}}
\newcommand{\eeqn}{\end{equation*}}
\newcommand{\bea}{\begin{eqnarray}}
\newcommand{\eea}{\end{eqnarray}}
\newcommand{\bean}{\begin{eqnarray*}}
\newcommand{\eean}{\end{eqnarray*}}
\newcommand{\ben}{\begin{enumerate}}
\newcommand{\een}{\end{enumerate}}
\newcommand{\bdefn}{\begin{defn}}
\newcommand{\edefn}{\end{defn}}
\newcommand{\bnote}{\begin{note}}
\newcommand{\enote}{\end{note}}
\newcommand{\bprop}{\begin{prop}}
\newcommand{\eprop}{\end{prop}}
\newcommand{\blem}{\begin{lem}}
\newcommand{\elem}{\end{lem}}
\newcommand{\bthm}{\begin{thm}}
\newcommand{\ethm}{\end{thm}}
\newcommand{\bconj}{\begin{conj}}
\newcommand{\econj}{\end{conj}}
\newcommand{\bconstr}{\begin{constr}}
\newcommand{\econstr}{\end{constr}}
\newcommand{\bpf}{\begin{proof}}
\newcommand{\epf}{\end{proof}}

\begin{document}
	\title{Codes with Combined Locality and Regeneration Having Optimal Rate, $d_{\text{min}}$ and Linear Field Size}
	
	\author{
		\IEEEauthorblockN{M. Nikhil Krishnan, Anantha Narayanan R., and P. Vijay Kumar, \it{Fellow}, \it{IEEE}}
		
		\IEEEauthorblockA{Department of Electrical Communication Engineering, Indian Institute of Science, Bangalore.  \\ Email: \{nikhilkrishnan.m, ananthnarayananr, pvk1729\}@gmail.com} 
		\thanks{P. Vijay Kumar is also a Visiting Professor at the University of Southern California. This research is supported in part by the National Science Foundation under Grant 1421848 and in part by an India-Israel UGC-ISF joint research program grant. M. Nikhil Krishnan would like to acknowledge the support of Visvesvaraya PhD Scheme for Electronics \& IT awarded by Department of Electronics and Information Technology, Government of India.}
	}
	\maketitle

	\begin{abstract}
		In this paper, we study vector codes with all-symbol locality, where the local code is either a Minimum Bandwidth Regenerating (MBR) code or a Minimum Storage Regenerating (MSR) code. In the first part, we present vector codes with all-symbol MBR locality, for all parameters, that have both optimal minimum-distance and optimal rate. These codes combine ideas from two popular codes in the distributed storage literature; Product-Matrix codes and Tamo-Barg codes. In the second part which deals with codes having all-symbol MSR locality, we follow a Pairwise Coupling Transform-based approach to arrive at optimal minimum-distance and optimal rate, for a range of parameters. All the code constructions presented in this paper have a low field-size that grows linearly with the code-length $n$.  
	\end{abstract}
	
	
	\section{Introduction}
	
In addition to the requirement of high storage-efficiency and reliability, there are two other important factors considered by Distributed Storage Systems (DSSs); (i) {\it repair bandwidth} incurred during a node-repair, and (ii) {\it repair degree}, which is the number of nodes contacted during a node-repair. Regenerating codes \cite{DimGodWuWaiRam} aim at minimizing the repair traffic, whereas codes with locality \cite{GopHuaSimYek} focus on reducing the number of nodes contacted during repair.
	
    In the regenerating code framework, a file of size $B$ symbols is encoded and stored across $n$ nodes, where each node stores $\alpha$ symbols. In the event of a node failure, the failed node can be regenerated by downloading $\beta\leq \alpha$ symbols each, from \textit{any} $d$ surviving nodes. Also, by accessing {\it any} $k$ nodes, the whole file can be retrieved. The parameters of a regenerating code are denoted by $((n,k,d),(\alpha,\beta),B)$. \cite{DimGodWuWaiRam} proves the existence of a trade-off between $\alpha$ (storage) and $d\beta$ (bandwidth) for given $n$, $k$, $d$, $\beta$ and file-size $B$. There are two codes belonging to the two extremal points in the trade-off, namely, \textit{Minimum Storage Regenerating (MSR) codes} and \textit{Minimum Bandwidth Regenerating (MBR) codes}, where $\alpha$ and $d\beta$ are minimized first respectively.
    
Under the codes-with-locality setting introduced by Gopalan et al. \cite{GopHuaSimYek}, an erased code-symbol can be repaired by accessing $r<k$ other symbols. This reduces the number of nodes accessed. The following minimum-distance bound is derived in \cite{GopHuaSimYek} for an $[n,k]$ linear code having $r$-locality: 
    \begin{equation}\label{eq:gopalan_bound}
    d_{\text{min}} \leq n-k- \lceil k/r \rceil +2
    \end{equation}
   The concept in \cite{GopHuaSimYek}, of having single parity check codes as local codes, is extended and stronger local codes are considered in \cite{KamPraLalKum}. Here local codes have a minimum-distance of at least $\delta$. The minimum-distance in this case, is upper bounded as:
	
	\begin{equation}\label{eq:kam_bound}
	d_{\text{min}} \leq n - k + 1 -(\Bigl\lceil\frac{k}{r}\Bigr\rceil - 1)(\delta - 1).
	\end{equation}
	  In \cite{TamBar}, Tamo and Barg provide a family of codes having locality that meets \eqref{eq:kam_bound}. 
	  
	 A natural question to ask at this point is, whether there exist codes which can simultaneously have a low repair bandwidth and a low repair degree. Kamath et al. \cite{KamPraLalKum} and Rawat et al. \cite{RawKoySilVis} answer this in the affirmative and present a new family of vector codes with locality, where the local codes are regenerating codes. These code constructions leverage the advantages of both regenerating codes (low repair bandwidth) and codes with locality (low repair degree). 
	  
	  In \cite{KamPraLalKum}, authors give minimum-distance bounds for general vector codes with locality and a tighter bound for the case when the local codes have Uniform Rank Accumulation (URA) property. Codes with MSR or MBR all-symbol locality and information-symbol locality, that meet the minimum-distance bound, are provided for various parameters. The field-size requirement is at least $O(n^2)$ for the all-symbol locality cases.  \cite{RawKoySilVis} presents an explicit construction of a vector code with MSR all-symbol locality, which requires a field-size exponential in $n$. In \cite{HashtagLocalRegen}, the authors construct a related family of vector codes with information-symbol locality, where the local codes are vector MDS codes with near-optimal bandwidth and small sub-packetization ($\alpha$) levels.
	  
	  {\it Our Results:} As a main result, we present a family of codes with all-symbol MBR locality, for {\it all} parameters. The construction is optimal with respect to the minimum-distance bound given in \cite{KamPraLalKum} and satisfies the rate-optimality property. Our results also include a family of codes having all-symbol MSR locality. These codes are shown to be optimal for a range of parameters. Both families of codes feature an $O(n)$ field-size, which is an improvement over prior work. 
	  
	\section{Preliminaries} 
	Let $[a,b]\triangleq\{a,a+1,\ldots,b\}$, $[a]\triangleq\{1,2,\ldots,a\}$. All the constructions are assumed to be linear and over $\mathbb{F}$, where $|\mathbb{F}|=q$.
	\subsection{Locality in Vector Codes}
	\begin{defn}{(Vector Codes)}
		A vector code $\mathcal{C}$ is a linear code over $\mathbb{F}$, with each codeword $\mathbf{c}\in\mathcal{C}$ taking the form:
		\begin{equation*}
		\mathbf{c}=(\mathbf{c}_0\ \mathbf{c}_1\ \ldots\ \mathbf{c}_{n-1}),
		\end{equation*} 
		where $\mathbf{c}_i\in \mathbb{F}^\alpha$, $0\leq i\leq (n-1)$, $\alpha\geq1$. i.e., each vector symbol $\mathbf{c}_i$ holds $\alpha$ scalar symbols.
	\end{defn}
    Consider the scalar code $\mathcal{C}_s$ of length $n\alpha$, obtained from $\mathcal{C}$, by expanding  each vector symbol $\mathbf{c}_i$ as $\alpha$ scalar symbols. Let $\mathbf{G}$ be a generator matrix for $\mathcal{C}_s$, where first $\alpha$ columns correspond to $\mathbf{c}_0$, the next $\alpha$ columns correspond to $\mathbf{c}_1$, and so on. Each set of $\alpha$ columns of $\mathbf{G}$ that corresponds to a vector symbol, is referred to as a {\it thick column}. The columns of $\mathbf{G}$ themselves will be referred to as {\it thin columns}. Hence, there are $\alpha$ thin columns within a thick column. Let $K$ denote the dimension of the code $\mathcal{C}_s$. The parameters of a vector code are denoted by $(n,K,d_{\text{min}},\alpha)$, where $d_{\text{min}}$ is the minimum-distance of $\mathcal{C}$, computed at the thick column level.
    
    For $S\subseteq[0,n-1]$, let $\mathcal{C}|_S$ denote the code obtained by puncturing (restricting) $\mathcal{C}$ to the set of thick columns $\{j:j\in\mathcal{S}\}$. In a similar manner, let $\mathbf{G}|_S$ be the restriction of the matrix $\mathbf{G}$ to the thick columns in $S$.
    
    \begin{defn}{($(r,\delta)$ Locality)}\label{def:r_delta_loc} For $i\in[0,n-1]$ and $\delta\geq 2$, the $i^{\text{th}}$ vector code symbol is said to have $(r,\delta)$ locality, if there exists an $S_i\subseteq[0,n-1]$ such that $i\in S_i$, $|S_i|\leq r+\delta-1$ and $d_{\text{min}}(\mathcal{C}|_{S_i})\geq \delta$. Any $\mathcal{C}|_{S_i}$ will be referred to as a local code.
    \end{defn}

    \begin{defn}{($(r,\delta)$ Information-Symbol Locality)} A vector code is said to have $(r,\delta)$ information-symbol locality if there exists $\mathcal{I}\subseteq[0,n-1]$ such that:
    	\bit
    	\item  $rank(\mathbf{G}|_{\mathcal{I}})=K$
    	\item For all $i\in\mathcal{I}$, $\mathbf{c}_i$ has $(r,\delta)$ locality.
    	\eit
    \end{defn}
 
    Furthermore, a vector code is said to have {\it$(r,\delta)$ all-symbol locality}, if for all $i\in[0,n-1]$, $\mathbf{c}_i$ has $(r,\delta)$ locality. If for a code having $(r,\delta)$ all-symbol locality, $S_i=S_j$ or $|S_i\cap S_j|=0$, for all $i\neq j$, $0\leq i,j\leq n-1$, then the code is said to have {\it disjoint locality}. All the code constructions presented in this paper have the disjoint locality property.
 
    \subsection{Codes with MBR/MSR Locality}
    
     A code with MSR or MBR locality \cite{KamPraLalKum} is an $(n,K,d_{\text{min}},\alpha)$ vector code with $(r,\delta)$ locality, where the local code is either MSR or MBR  with parameters $((n_\ell,r,d),(\alpha,\beta),K_\ell)$. Here $n_\ell\triangleq r+\delta-1$ and $K\geq K_\ell$. Let the local code be denoted by $\mathcal{C}_\text{loc}$, with an associated generator matrix $G_\text{loc}$. Both MSR and MBR codes belong to a class of Uniform Rank Accumulation (URA) codes, where there exists a non-increasing sequence of $n_\ell$ non-negative integers  $\{a_1,a_2,\ldots,a_{n_\ell}\}$ with the following properties (i) $a_1=\alpha$ (ii) $\text{rank}(\mathbf{G}_\text{loc}|_{\mathcal{I}}) = \sum_{j=1}^{i}a_j$, for all $\mathcal{I}\subseteq[0,n_\ell-1]$ such that $|\mathcal{I}|=i$. The sequence $\{a_i, i \in [n_\ell]\}$ is referred to as the rank profile of the vector code $\mathcal{C}_{\text{loc}}$.
     
     The rank profile of an $((n_\ell,r,d),(\alpha,\beta),K_\ell)$ MSR code is given by (see for example, \cite{ShaRasKumRam_rbt}):
    \begin{equation}\label{eq:msr_rank_profile}
    a_i = \left\{\begin{array}{lr}
    \alpha &  1\leq i\leq r\\
    0 & (r+1)\leq i\leq n_\ell
    \end{array}\right..
    \end{equation}
    For the MBR code, rank profile \cite{ShaRasKumRam_rbt} is as follows:
    
    \begin{equation}\label{eq:mbr_rank_profile}
    a_i = \left\{\begin{array}{lr}
    \alpha-(i-1)\beta &  1\leq i\leq r\\
    0 & (r+1)\leq i\leq n_\ell
    \end{array}\right..
    \end{equation}
   
   Define $b_{i+jn_\ell}\triangleq a_i$, where $1\leq i\leq n_\ell$ and $j\geq 0$. Let 
   \begin{equation}\label{eq:leading_sum}
   P(s)=\sum_{i=1}^sb_i:\ s\geq1.
   \end{equation}
   For $1\leq x\leq K$, let $P^{(\text{inv})}(x)=y$, where $y$ is the smallest integer such that $P(y)\geq x$. From \cite{KamPraLalKum}, we have the minimum-distance upper bound:
   \begin{equation}\label{eq:min_distance_bound}
   	d_{\text{min}} \leq n -P^{(\text{inv})}(K) + 1
   	\end{equation}
    A code satisfying \eqref{eq:min_distance_bound} with equality is defined to be {\it rate-optimal} \cite{KamPraLalKum}, if $K=P(s)$, for some $s\geq1$.
    For a code with MSR locality, one can simplify \eqref{eq:min_distance_bound} using \eqref{eq:msr_rank_profile} to obtain (\cite{KamPraLalKum}, \cite{RawKoySilVis}) :
    \begin{equation}\label{eq:msr_locality_bound}
    d_{\text{min}}\leq n-\bigg\lceil\frac{K}{\alpha}\bigg\rceil+1-\bigg(\bigg\lceil\frac{K}{\alpha r}\bigg\rceil-1\bigg)(\delta-1).
    \end{equation}

   \subsection{Product-Matrix (PM) MBR Codes}\label{recapSection}
	
	PM MBR codes \cite{RasShaKum_pm} exist for all $n,k,d$ and $\beta=1$. At the MBR point, $\alpha=d\beta=d$ and $B= kd\beta - {k \choose 2}\beta=kd-{k\choose 2}$. Consider a symmetric $(d\times d)$ message matrix $\mathbf{M}$: 
	\begin{align}
	\mathbf{M} &= 
	\begin{bmatrix}
	\bf{S} & \bf{T}\\           
	$\bf{T}$^t & \bf{0} 
	\end{bmatrix},\label{eq:PMmessagematrix}
	\end{align} 
	where $\bf{S}$ is a symmetric $k\times k$ matrix which can hold ${k+1\choose 2}$ independent scalar message symbols, $\bf{T}$ is a $k\times (d-k)$ matrix which can hold $k(d-k)$ independent scalar message symbols. Note that the quantities ${k+1\choose 2}$ and $k(d-k)$ add up to $B$. For $0\leq i\leq n-1$, the $i^{\text{th}}$ vector code symbol (to be stored at node-$i$) is given by $\mathbf{c}_i=\boldsymbol{\psi}_i\mathbf{M}$, where $\boldsymbol{\psi}_i$ takes the form: $\boldsymbol{\psi}_i=[1\  \alpha_i\  \alpha_i^2\ \ldots \ \alpha_i^{d-1} ]$.
	Here $\alpha_i$'s are chosen to be distinct. The PM MBR construction can be formulated as a polynomial evaluation code as described in the example below.

\begin{example}\label{eg:MBR_poly_eval}
	Let $n=5$, $k=3$ and $d=4$. The message matrix, $\mathbf{M}$ is given by:
	\begin{align}\nonumber
	\mathbf{M}= 
	\begin{bmatrix}
	m_{0,0} & m_{0,1} & m_{0,2} & m_{0,3}\\           
	m_{0,1} & m_{1,1} & m_{1,2} & m_{1,3}\\
	m_{0,2} & m_{1,2} & m_{2,2} & m_{2,3}\\
	m_{0,3} & m_{1,3} & m_{2,3} & 0
	\end{bmatrix}.
	\end{align}
    The symbols stored at node-$i$,
	$\boldsymbol{c}_i=[1\ \ \alpha_i\ \ \alpha_i^2\ \ \alpha_i^3]\mathbf{M}$
	can be alternatively viewed as the evaluation of a vector of polynomials $[{m}_0(x)\ {m}_1(x)\ {m}_2(x)\ {m}_3(x)]$ at $\alpha_i$, where ${m}_0(x)\triangleq m_{0,0}+m_{0,1}x+m_{0,2}x^2+m_{0,3}x^3$, ${m}_1(x)\triangleq m_{0,1}+m_{1,1}x+m_{1,2}x^2+m_{1,3}x^3$, ${m}_2(x)\triangleq m_{0,2}+m_{1,2}x+m_{2,2}x^2+m_{2,3}x^3$ and ${m}_3(x)\triangleq m_{0,3}+m_{1,3}x+m_{2,3}x^2$.
\end{example}

\subsection{Tamo-Barg (TB) Codes}\label{sec:tb_code_recap}

In this section, we summarize an $[n,k,d_{\text{min}}]$ scalar (i.e., $\alpha=1$) linear code construction with $(r,\delta)$ all-symbol locality, introduced in \cite{TamBar}, where $(r+\delta-1)|n$,  $n_\ell\triangleq r+\delta-1$, $k\leq \frac{nr}{n_\ell}$. We refer to this as the {\it Tamo-Barg} (TB) code. The construction is minimum-distance optimal with respect to \eqref{eq:kam_bound}. 

Let $n|(q-1)$, $\nu\triangleq\frac{n}{n_{\ell}}$, $\gamma$ be a primitive $n^{\text{th}}$ root of unity and $\mathcal{A}\triangleq\{1,\gamma,\gamma^2,\ldots,\gamma^{n-1}\}\subset\mathbb{F}$. Each codeword  of the TB code will correspond to $n$ evaluations of some polynomial $M(x)$ at the points in $\mathcal{A}$. $M(x)$ belongs to a $k$-dimensional subspace $\mathcal{M}$, of the vector space of polynomials over $\mathbb{F}$ with degree at most $(n-1)$. In the following, we describe the construction of $\mathcal{M}$.


Consider the partition of $\mathcal{A}$ into the multiplicative subgroup $\mathcal{A}^{(0)}=\{1,\gamma^{\nu},\ldots,(\gamma^{\nu})^{n_{\ell}-1}\}$ and its cosets $\mathcal{A}^{(i)}=\gamma^i\mathcal{A}^{(0)}$, for $0\leq i\leq \nu-1$. For each $\mathcal{A}^{(i)}$, let $f^{(i)}$ denote the {\it annihilating polynomial}, i.e., $f^{(i)}(x)=\prod_{\theta\in \mathcal{A}^{(i)}}(x-\theta)=(x^{n_l}-(\gamma^i)^{n_l})$. Clearly, $f^{(i)}(x)$'s are pairwise co-prime. Let $F(x)\triangleq\prod_{i=0}^{\nu-1}f^{(i)}(x)=(x^{n}-1)$. By applying the Chinese Remainder Theorem (CRT), one can observe the following isomorphism:
\begin{equation}
\mathbb{F}(x)/(F)\approxeq \mathbb{F}(x)/(f^{(0)}) \times \mathbb{F}(x)/(f^{(1)}) \times \ldots \times \mathbb{F}(x)/(f^{(\nu-1)}), \label{eq:CRT_isomorphism}
\end{equation}
where $\mathbb{F}[x]$ is the ring of polynomials over $\mathbb{F}$.

Consider $\nu$ polynomials of degree at most $(r-1)$, $\tilde{m}^{(i)}(x)$ for $0\leq i\leq \nu-1$. Think of them as a vector of $\nu$ polynomials belonging to a vector space $\mathcal{\tilde{M}}^{(\nu)}$, of dimension $r\nu$. Applying CRT, one can find the unique polynomial $\tilde{M}(x)$ of degree at most $(n-1)$ such that:

\begin{equation}\label{eq:poly_lifting}
\tilde{M}(x) \mod f^{(i)}(x)=  \tilde{m}^{(i)}(x),
\end{equation} 
for all $0\leq i\leq \nu-1$. The process of obtaining $\tilde{M}(x)$ from $\tilde{m}^{(i)}(x)$'s is termed as {\it polynomial lifting}. There exist (\cite{hierarchical}) $\{e^{(i)}(x)\}_{i=0}^{\nu-1}$, where each $e^{(i)}(x) \in \mathbb{F}[x]$, has degree $n_\ell(\nu-1)$ and satisfies:
\[
e^{(i)}(x)\mod f^{(j)}(x) = \left\{\begin{array}{lr}
1 &  j=i\\
0 & j\neq i
\end{array}\right..
\]
Moreover, each $e^{(i)}(x)$ takes the form: $e^{(i)}(x)=\sum_{j=0}^{\nu-1}e^{(i)}_j(x^{n_{\ell}})^j$. Clearly, $\tilde{M}(x)=\sum_{i=0}^{\nu-1}\tilde{m}^{(i)}(x)e^{(i)}(x)$.

\begin{note}\label{remark:eval_match}
	From the definition of $e^{(i)}(x)$, it is easy to see that:
	\begin{equation*}\label{eq:eval_match}
	\tilde{m}^{(i)}(a)=\tilde{M}(a)\ \  \forall a\in \mathcal{A}^{(i)}
	\end{equation*}
\end{note}

We describe the CRT-based TB code with the help of the following example.

\begin{example}\normalfont\label{eg:TB_code_example}
	Consider the parameters $n=15$, $k=6$, $r=3$, $\delta=3$ and let $|\mathbb{F}|=16$. Here $n_l=5$, $\nu=3$, $\tilde{m}^{(i)}(x)=\tilde{m}^{(i)}_{0}+\\
	\tilde{m}^{(i)}_{1}x+\tilde{m}^{(i)}_{2}x^2, \mathcal{A}^{(i)}=\gamma^i\{1,\gamma^3,\gamma^6,\gamma^9,\gamma^{12}\},e^{(i)}(x)=e^{(i)}_{0}+e^{(i)}_{1}x^5+e^{(i)}_{2}x^{10}$, $0\leq i\leq 2$. $\tilde{M}(x)$ is given by:
	
	\begin{eqnarray*}
	\tilde{M}(x) & = & 	\sum_{i=0}^{2}\tilde{m}^{(i)}(x)e^{(i)}(x) \\
	& = & \sum_{i=0}^{2}\tilde{m}^{(i)}(x)e^{(i)}_{0}+ \bigg[\sum_{i=0}^{2}\tilde{m}^{(i)}(x)e^{(i)}_{1}\bigg]x^5 \\
	&  & + \bigg[\sum_{i=0}^{2}\tilde{m}^{(i)}(x)e^{(i)}_{2}\bigg]x^{10}
	\end{eqnarray*}
	Let $\tilde{M}(x)\triangleq\sum_{j=0}^{n-1}\tilde{M}_jx^j$. For $0\leq j\leq (n-1)$, $0\leq a\leq \nu-1$ and $0\leq b\leq n_\ell-1$,  let $j=an_\ell+b$. We have:
	\begin{equation}\label{eq:crt_map}
	\tilde{M}_j  =
	\begin{cases} 
	\sum_{s=0}^{\nu-1}e^{(s)}_{a}\tilde{m}^{(s)}_b & b \leq (r-1), \\
	0 & \text{otherwise.}
	\end{cases}
	\end{equation}
	
	Let $\mathcal{\tilde{M}}$ denote the vector space of all possible $\tilde{M}(x)$'s. The dimension of $\mathcal{\tilde{M}}$ equals $r\nu=9$. This follows from CRT, as the vector space $\mathcal{\tilde{M}}^{(\nu)}$ has dimension $r\nu$. Let $\mathcal{T}$ indicate the collection of indices $t\in[0,n-1]$, for which there exists an  $\tilde{M}(x)\in\mathcal{\tilde{M}}$ with  $\tilde{M}_t\neq 0$ (i.e., the set of $j$'s for which first case in \eqref{eq:crt_map} is true). For our example, $\mathcal{T}=\{0,1,2,5,6,7,10,11,12\}$. As dimension of $\mathcal{\tilde{M}}$ equals the quantity $|\mathcal{T}|$, $\mathcal{\tilde{M}}$ is nothing but the vector space spanned by the set of monomials $\{x^t:t\in \mathcal{T}\}$.
	
	Now we shall see how to construct the required code $\mathcal{C}_{\text{TB}}$ with $(r=3,\delta=3)$ locality. Consider the subspace $\mathcal{M}$ of $\mathcal{\tilde{M}}$, with dimension $k=6$, obtained as follows. Let $\mathcal{M}$ be the subspace containing all $\tilde{M}(x)\in\mathcal{\tilde{M}}$ for which $\tilde{M}_j=0$ for the $(r\nu-k)=3$ largest indices in $\mathcal{T}$. i.e., $\mathcal{M}$ is the  vector space spanned by the set of monomials $\{x^t:t\in \{0,1,2,5,6,7\}\}$. Each codeword in $\mathcal{C}_{\text{TB}}$ will be evaluations of an $M(x)\in\mathcal{M}$ over the set $\mathcal{A}$. From Remark \ref{remark:eval_match}, when restricted to the points in $\mathcal{A}^{(i)}$, evaluations of $M(x)$ can be seen as evaluations of a lower degree polynomial with degree at most $(r-1)$. This essentially implies locality. As for the minimum-distance, the largest degree possible for $M(x)$, is $7$. Hence the number of roots possible are at most $7$, at $n=15$ evaluation points. Thus $d_{\text{min}}(\mathcal{C}_{\text{TB}})\geq 15-7=8$. This matches the upper bound in \eqref{eq:kam_bound}. 
	
	\begin{note}\label{remark:MDS_remark}
		It is known that each local code in the TB code is an MDS code of length $n_\ell$ and dimension $r$. In other words,  $\{M(x) \mod f^{(i)}(x): M(x)\in \mathcal{M}\}$	is in fact the set of all polynomials having degree at most $(r-1)$, $\forall i\in[0,\nu-1]$.
     \end{note}
\end{example}	

\subsection{Pairwise Coupling Transform (PCT) to Construct MSR Codes}

There is a sequence of works \cite{LiTangTianEnablingAllNodeRepair16}, \cite{YeBargPCTCode}, \cite{SasVajhaKum16}, \cite{TianLiTangISIT17} which share a certain Pairwise Coupling Transform (PCT) idea that can be used to obtain high-rate MSR codes from scalar MDS codes. We summarize the scheme as follows.

Let the MSR code parameters be $((n=st,k=s(t-1),d=(n-1)),(\alpha = s^t, \beta = s^{t-1}), B=k\alpha)$, where $s\geq 2$, $t \geq 2$. The $n$ nodes are indexed using $(x,y)$, where $0\leq x\leq s-1$, $1\leq y\leq t$. Each scalar symbol $A$ in an MSR codeword is indexed by a triplet denoted by: $(x,y,\underline{z})$, where $x \in \mathbb{Z}_s, y \in [t], \underline{z} \in \mathbb{Z}_s^t$. Here $\mathbb{Z}_s$ denotes the integers modulo $s$. The $(x,y)$ pair determines the node, while $\underline{z}$ determines the position of symbol within a node. Let $z_i$, for $1\leq i\leq t$, denote the $i^\text{th}$ element of $\underline{z}$ and $\underline{z}(x,y) \triangleq (z_1,z_2,\dots,z_{y-1},x,z_{y+1},\dots,z_t).$ In order to obtain the MSR code symbols $\{A(x,y,\underline{z})\}$, we initially populate every $(x,y,\underline{z})$ coordinate with a code-symbol $B(x,y,\underline{z})$. Here, for every fixed $\underline{z}$, $\{B(x,y,\underline{z})\}_{x\in[0,s-1],y\in[t]}$ corresponds to an independent layer of $[n,k]$ MDS code.
The coupled symbols, $A(x,y,\underline{z})$ can be written in terms of a $2\times 2$ coupling matrix, $C$ and uncoupled symbols, $B(x,y,\underline{z})$ as:

\begin{equation*}\label{eq:coupling}		
\begin{bmatrix}
A(x,y,\underline{z}) \\
A(z_y,y,\underline{z}(x,y))
\end{bmatrix}	
=
C
\begin{bmatrix}
B(x,y,\underline{z}) \\
B(z_y,y,\underline{z}(x,y))
\end{bmatrix}		
,~\forall x \neq z_y
\end{equation*}
\begin{equation*}
A(x,y,\underline{z}) = B(x,y,\underline{z}) ,~\forall x = z_y.
\end{equation*}
Here $C$ is chosen in such a way that any two out of the four (two coupled + two uncoupled) symbols will be sufficient to obtain the other two symbols. Let $\mathcal{P}\subseteq\{(x,y):\ x\in[0,s-1],y\in[t]\}$. Define for $1\leq i\leq t$, $\mathcal{P}_{i}=\{x:(x,i)\in\mathcal{P}\}$. We derive the following lemma.

\begin{lem}\label{lem:coupling_lemma}
	If $A(x,y,\underline{z})=0$ $\forall(x,y)\in \mathcal{P}$ and $\forall \underline{z} \in \mathbb{Z}_s^t$, then $B(x',y',\underline{z}')=0$, $\forall x',\underline{z}':x'\in \mathcal{P}_{y'}, z'_{y'}\in \mathcal{P}_{y'}$.
\end{lem}
\begin{proof}
	If $z'_{y'}=x'$ and $x'\in\mathcal{P}_{y'}$, clearly, $B(x',y',\underline{z}')=A(x',y',\underline{z}')=0$. If $z'_{y'}\neq x'$, we have the coupled symbols $A(x',y',\underline{z}')$ and $A(z'_{y'},y',\underline{z}'(x',y'))$. As $A(x',y',\underline{z}')=A(z'_{y'},y',\underline{z}'(x',y'))=0$ (follows from $\{x',z'_{y'}\}\in\mathcal{P}_{y'}$ assumption), $B(x',y',\underline{z}')=B(z'_{y'},y',\underline{z}'(x',y'))=0$.
\end{proof}
\begin{cor}\label{cor:minimum_distance_layer0}
	If $A(x,y,\underline{z})=0$ $\forall(x,y)\in \mathcal{P}$ and $\forall \underline{z} \in \mathbb{Z}_s^t$, there exists a $\underline{z}'\in \mathbb{Z}_s^t$ such that  $B(x,y,\underline{z})=0$ $\forall(x,y)\in \mathcal{P}$. 
\end{cor}

\section{Codes with MBR Locality }\label{sec:MBR_local_regen}
In this section, we present a family of codes with MBR $(r,\delta)$ all-symbol  locality, which is optimal with respect to \eqref{eq:min_distance_bound}. In contrast to the existing code constructions, these codes require a low field-size of $O(n)$. The construction is based on Product-Matrix MBR codes \cite{RasShaKum_pm} and Tamo-Barg codes \cite{TamBar}. 

{\it Parameters}: Let $\mathcal{C}$ be an $(n,K,d_{\text{min}},\alpha)$ vector code, with $((n_\ell,r,d),(\alpha,\beta),K_\ell)$ local MBR codes. We consider the disjoint locality case and hence have $n_\ell|n$,  $\nu\triangleq\frac{n}{n_\ell}$. Let $K$ be such that $P(P^{(\text{inv})}(K))=K$ (i.e., rate optimal \cite{KamPraLalKum}), where $P(.)$ is as defined in \eqref{eq:leading_sum}.

  Let $\hat{\mathbf{M}}^{(i)}$ be the $d\times d$ message matrix (under the PM MBR framework) corresponding to local code-$i$, for $0\leq i\leq \nu-1$. We consider MBR codes as polynomial evaluation codes, as seen in Example \ref{eg:MBR_poly_eval}. Our aim here is to introduce dependencies across the matrices $\{\hat{\mathbf{M}}^{(i)}\}$ so as to obtain the desired $K$ (with $K_\ell\leq K\leq  \nu K_\ell$) and the optimal minimum-distance. Assume $n|(q-1)$ and let $\mathcal{A}^{(i)}\triangleq\gamma^i\{1,\gamma^\nu,(\gamma^\nu)^2,\ldots,(\gamma^\nu)^{n_\ell-1}\}$ denote the set of evaluation points for the $i^{\text{th}}$ MBR local code, where $\gamma$ is a primitive $n^{\text{th}}$ root of unity. The message matrix $\hat{\mathbf{M}}^{(i)}$ is given by:
  
  \begin{align}\nonumber
  \hat{\mathbf{M}}^{(i)} = 
  \begin{bmatrix}
  \hat{\mathbf{S}}^{(i)} & (\hat{\mathbf{T}}^{(i)})^T\\
  \hat{\mathbf{T}}^{(i)} & 0 \\
  \end{bmatrix}\nonumber
  \end{align}
  where,

  \begin{align}\nonumber
  \hat{\mathbf{S}}^{(i)} = 
  \begin{bmatrix}
  \hat{m}^{(i)}_{0,0} & \hat{m}^{(i)}_{0,1} & \dots & \hat{m}^{(i)}_{0,r-1} \\         
  \hat{m}^{(i)}_{0,1} & \hat{m}^{(i)}_{1,1} & \dots & \hat{m}^{(i)}_{1,r-1} \\
  \vdots & \vdots & \ddots & \vdots\\
  \hat{m}^{(i)}_{0,r-1} & \hat{m}^{(i)}_{1,r-1} & \dots & \hat{m}^{(i)}_{r-1,r-1}\\
  \end{bmatrix}, \nonumber
  \end{align}

  \begin{align}\nonumber
  \hat{\mathbf{T}}^{(i)} = 
  \begin{bmatrix}
  \hat{m}^{(i)}_{0,r} & \hat{m}^{(i)}_{1,r} & \dots & \hat{m}^{(i)}_{r-1,r} \\         
  \hat{m}^{(i)}_{0,r+1} & \hat{m}^{(i)}_{1,r+1} & \dots & \hat{m}^{(i)}_{r-1,r+1} \\
  \vdots & \vdots & \ddots & \vdots\\
  \hat{m}^{(i)}_{0,d-1} & \hat{m}^{(i)}_{1,d-1} & \dots & \hat{m}^{(i)}_{r-1,d-1}\\
  \end{bmatrix}. \nonumber
  \end{align}
  
  For $0\leq i\leq \nu-1$, the $i^\text{th}$ MBR local code is obtained by evaluating $[\hat{m}^{(i)}_0(x)\ \hat{m}^{(i)}_1(x)\ \dots \ \hat{m}^{(i)}_{d-1}(x)]$ at the $n_\ell$ evaluation points given by $\mathcal{A}^{(i)}\triangleq\gamma^i\{1,\gamma^\nu,(\gamma^\nu)^2,\ldots,(\gamma^\nu)^{n_\ell-1}\}$, where $\hat{m}_j^{(i)}(x) \triangleq \sum_{t=0}^{d-1}\hat{m}^{(i)}_{j,t}x^t$. Let the $d$ columns of any message matrix $\hat{\mathbf{M}}^{(i)}$ be indexed by $j$, where $0\leq j\leq d-1$. Since $M$ is symmetric, we replace the notation $\hat{m}^{(i)}_{x,y}$ with $\hat{m}^{(i)}_{y,x}$ whenever $x>y$. Also,  $\hat{m}^{(i)}_{x,y}\triangleq 0$, when $r\leq x,y\leq d-1$.

  Fix a column $j$, for all the $\nu$ message matrices. Thus we have $\nu$ polynomials; $\hat{m}^{(0)}_j(x), \hat{m}^{(1)}_j(x),$
  $ \ldots, \hat{m}^{(\nu-1)}_j(x)$.
  We shall perform polynomial lifting to arrive at the polynomial $\hat{M}_j(x)$, which has the property (similar to that of $\tilde{M}(x)$ stated in \eqref{eq:poly_lifting}): $\hat{M}_j(x) \mod f^{(i)}(x) =  \hat{m}_j^{(i)}(x)$ for all $0\leq i\leq \nu-1$. Note that the vector space $\hat{\mathcal{M}}_j$ of all possible $\hat{M}_j(x)$'s has its dimension as follows:
  \begin{equation*}\label{eq:lifted_poly_vector_space_dimension}
  \text{dim}(\hat{\mathcal{M}}_j)  =
  \begin{cases} 
  d\nu & 0\leq j \leq (r-1), \\
  r\nu & r\leq j \leq (d-1).
  \end{cases}
  \end{equation*}
  This is because we have not assumed any dependencies across the column $j$ of the $\nu$ message matrices, to start with. Let $\hat{M}_j(x)\triangleq\sum_{t=0}^{n-1}\hat{M}_{j,t}x^t$. For $0\leq t\leq (n-1)$, $0\leq a\leq \nu-1$ and $0\leq b\leq n_\ell-1$,  let $t=an_\ell+b$. For $0\leq j\leq (r-1)$, we have:
   \begin{equation}\label{eq:crt_map1_MBR}
   \hat{M}_{j,t}  =
   \begin{cases} 
   \sum_{s=0}^{\nu-1}e^{(s)}_{a}\hat{m}^{(s)}_{j,b} & b \leq (d-1), \\
   0 & \text{otherwise.}
   \end{cases}
   \end{equation}
   Similarly, for $r\leq j\leq (d-1)$, we have:
   \begin{equation}\label{eq:crt_map2_MBR}
   \hat{M}_{j,t}  =
   \begin{cases} 
   \sum_{s=0}^{\nu-1}e^{(s)}_{a}\hat{m}^{(s)}_{j,b} & b \leq (r-1), \\
   0 & \text{otherwise.}
   \end{cases}
   \end{equation}
   
   Let $\mathcal{T}_j$ indicate the collection of indices $t\in[0,n-1]$, for which there exists an  $\hat{M}_j(x)\in\mathcal{\hat{M}}_j$ with  $\hat{M}_{j,t}\neq 0$ (i.e., the collection of $t$'s for which first case in \eqref{eq:crt_map1_MBR} or \eqref{eq:crt_map2_MBR} is true). Therefore,   
    \begin{equation}\label{eq:nonzero_degrees}
  \mathcal{T}_j  =
   \begin{cases} 
   \cup_{z=0}^{\nu-1}\{zn_\ell,zn_\ell+1,\ldots,zn_\ell+d-1\} & 0\leq j\leq  (r-1), \\
   \cup_{z=0}^{\nu-1}\{zn_\ell,zn_\ell+1,\ldots,zn_\ell+r-1\} & r\leq j\leq  (d-1)
   \end{cases}.
   \end{equation}
   Similar to the case in Example \ref{eg:TB_code_example}, $\mathcal{\hat{M}}_j$ is precisely the space spanned by the set of polynomials $\{x^t:t\in \mathcal{T}_j\}$.
   
   {\it Construction for Code with MBR Locality}: Let $K=aK_\ell +b$, where $0\leq a\leq \nu-1$, $1\leq b\leq K_\ell$. If $a=0$, $b=K_\ell$. For the last $d-P^{(\text{inv})}(b)$ columns, i.e., $P^{(\text{inv})}(b)\leq j\leq d-1$, consider the subspace $\mathcal{M}_j$ of $\hat{\mathcal{M}}_j$, spanned by monomials of degree at most $an_\ell+P^{(\text{inv})}(b)-1$. For columns $0\leq j\leq P^{(\text{inv})}(b)-1$, take $\mathcal{M}_j$ to be the space spanned by monomials of degree at most $an_\ell+d-1$. This is essentially equivalent to introducing dependencies for each column $j$ across all the $\nu$ message matrices (as we have seen in Example \ref{eg:TB_code_example}). Let $\{\mathbf{M}^{(i)}\}_{i=0}^{\nu-1}$ be the collection of message matrices obtained after introducing dependencies. $\mathcal{C}$ is obtained by individually evaluating (as in Example \ref{eg:MBR_poly_eval}) each message matrix ${\mathbf{M}}^{(i)}$ at the respective evaluation points in $\mathcal{A}_i$. Note that by Remark \ref{remark:eval_match}, each codeword of the length-$n$ scalar code, obtained by restricting $\mathcal{C}$ to column $j$, is nothing but the $n$ evaluations of a polynomial in $\mathcal{M}_j$.
   
   \begin{claim}
   	$\mathcal{C}$ is an $(n,K,d_{\text{min}},\alpha)$ vector code with $((n_\ell,r,d),(\alpha,\beta),K_\ell)$-MBR locality, where $d_\text{min}$ meets the upper bound \eqref{eq:min_distance_bound}.
   \end{claim}
   
   \begin{proof}{(outline)}
   	We need to prove three things here; (i) the local code is an MBR code with parameters 
   	$((n_\ell,r,d),(\alpha,\beta),K_\ell)$ (ii) $\mathcal{C}$ has scalar dimension $K$ and (iii) $\mathcal{C}$ is $d_\text{min}$-optimal. 
   	
   	{\it MBR locality}: From Remark \ref{remark:MDS_remark}, one can infer that the space of all ${\mathbf{M}}^{(i)}$'s is same as the space of all $\hat{\mathbf{M}}^{(i)}$, which is a subspace of the space of $d\times d$ symmetric matrices and has dimension $K_\ell$. Hence each local code will be an $((n_\ell,r,d),(\alpha,\beta),K_\ell)$-MBR code.
   	
   	{\it Scalar dimension $K$}:  All the columns $j\in[0,d-1]$ give rise to lifted polynomials of degree at most $(an_\ell+d-1)$. In other words $ \hat{M}_{j,t}$'s must be zeros for all $j\in[0,d-1]$ and $t>an_\ell+d-1$ in \eqref{eq:crt_map1_MBR} and \eqref{eq:crt_map2_MBR}. It can be verified that, as $\hat{m}^{(s)}_{j,i}=\hat{m}^{(s)}_{i,j}$ (symmetry of message matrices), this results in a total of $(\nu-a-1)K_\ell$ dependencies. The lifted polynomials arising from columns $j\in [P^{(\text{inv})}(b),d-1]$ are further constrained to a degree of at most  $(an_\ell+P^{(\text{inv})}(b)-1)$.
   	
   	From \eqref{eq:mbr_rank_profile}, we have:
   	\begin{equation}
   	\sum_{z=0}^{r-P^{(\text{inv})}(b)-1}(d-r+1+z)+\sum_{z=0}^{P^{(\text{inv})}(b)-1}(d-z)=K_\ell.
   	\end{equation}
   	Note that as $K$ is chosen to be such that $P(P^{(\text{inv})}(K))=K$, this also means that $P(P^{(\text{inv})}(b))=b$.
   	Hence from \eqref{eq:mbr_rank_profile} and \eqref{eq:leading_sum}, we can infer that $\sum_{z=0}^{P^{(\text{inv})}(b)-1}(d-z)=b$. Because of the symmetric nature of message matrices, the number of additional dependencies (they are already constrained to a maximum degree of $an_\ell+d-1$) that need to be introduced to constraint the last $d-P^{(\text{inv})}(b)$ columns to a degree of at most $an_\ell+P^{(\text{inv})}(b)-1$ can be verified to be precisely $\sum_{z=0}^{r-P^{(\text{inv})}(b)-1}(d-r+1+z)=K_\ell-b$. Thus dimension of $\mathcal{C}$ is $\nu K_\ell-(\nu-a-1)K_\ell-\sum_{z=0}^{r-P^{(\text{inv})}(b)-1}(d-r+1+z)=\nu K_\ell-(\nu-a-1)K_\ell-(K_\ell-b)=aK_\ell+b=K$.

   	{\it Minimum-distance optimality}: From \eqref{eq:min_distance_bound}, 	$d_{\text{min}} \leq n -P^{(\text{inv})}(K) + 1= n-(an_\ell+P^{(\text{inv})}(b))+1$. For the scalar code (polynomial evaluation code) obtained by restricting $\mathcal{C}$ to any of the columns $j\in [P^{(\text{inv})}(b),d-1]$, the largest degree of the underlying polynomial is restricted to $(an_\ell+P^{(\text{inv})}(b)-1)$ by design. Hence for the columns  $j\in [P^{(\text{inv})}(b),d-1]$, the scalar minimum-distance is at least $n-(an_\ell+P^{(\text{inv})}(b)-1)=n-(an_\ell+P^{(\text{inv})}(b))+1$. If for some choice of $\{\mathbf{M}^{(i)}\}_{i=0}^{\nu-1}$, all the columns in the range $j\in [P^{(\text{inv})}(b),d-1]$ yield all-zero codewords, this essentially means the message matrices $\{\mathbf{M}^{(i)}\}_{i=0}^{\nu-1}$, when restricted to these columns are all-zero matrices. As all the message matrices are symmetric, the last $d-P^{(\text{inv})}(b)$ rows will also be zeros for all the message matrices. Thus, for the columns in the range $[0,P^{(\text{inv})}(b)-1]$, lifted polynomials lie in the span of $\{x^t:t\in \mathcal{T}'_j\}$, where $\mathcal{T}'_j=\cup_{z=0}^{\nu-1}\{zn_\ell,zn_\ell+1,\ldots,zn_\ell+P^{(\text{inv})}(b)-1\}$. However by design, the degree is at most $(an_\ell+d-1)$, for these polynomials. Hence the maximum degree possible is $(an_\ell+P^{(\text{inv})}(b)-1)$. Thus, if the last $d-P^{(\text{inv})}(b)$ columns give rise to all-zero codewords, the first $P^{(\text{inv})}(b)$ columns will give scalar codewords with a minimum-distance of at least $n-(an_\ell+P^{(\text{inv})}(b))+1$. This proves the minimum-distance optimality of $\mathcal{C}$.
\end{proof} 		
Thus, we have the following theorem.    
    \begin{thm}
    	Linear field-size constructions exist for minimum-distance optimal, rate-optimal $(n,K,d_{\text{min}},\alpha)$ vector codes, with $((n_\ell,r,d),(\alpha,\beta),K_\ell)$ local MBR codes, where $1\leq r\leq d\leq (n_\ell-1)$ and $n_\ell|n$.
    \end{thm}
    
\begin{example}\label{eg:mbr_locality}
	\normalfont
	
	Let $n=12$, $n_\ell=6$, $r=3$, $d=4$, $K_\ell=9$, $K=13$, $\beta=1$.  Hence we have $\delta=n_\ell-r+1=4$, $P(s)=(4\ 7\ 9\ 9\ 9\ 9\ 13\ 16\ 18\ 18\ 18\ 18)$, indexed over $1\leq s\leq n=12$. Note that $P^{(\text{inv})}(K)=7$ and thus $d_{\min}\leq 6$. Let $\hat{\mathbf{M}}^{(i)}$, $0\leq i\leq \nu-1=1$, the MBR message matrix corresponding to the $i^\text{th}$ local MBR code, be as given below. The $(x,y)^{\text{th}}$ element of $\hat{\mathbf{M}}^{(i)}$ is denoted by 	$\hat{m}^{(i)}_{x,y}$. Note that as $\hat{\mathbf{M}}^{(i)}$ is a symmetric matrix,  $\hat{m}^{(i)}_{x,y}=\hat{m}^{(i)}_{y,x}$. 
	\begin{align}\nonumber
	\hat{\mathbf{M}}^{(i)} = 
	\begin{bmatrix}
	\hat{m}^{(i)}_{0,0} & \hat{m}^{(i)}_{0,1} & \hat{m}^{(i)}_{0,2} & \hat{m}^{(i)}_{0,3}\\         
	\hat{m}^{(i)}_{0,1} & \hat{m}^{(i)}_{1,1} & \hat{m}^{(i)}_{1,2} & \hat{m}^{(i)}_{1,3} \\
	\hat{m}^{(i)}_{0,2} & \hat{m}^{(i)}_{1,2} & \hat{m}^{(i)}_{2,2} & \hat{m}^{(i)}_{2,3}\\
	\hat{m}^{(i)}_{0,3} & \hat{m}^{(i)}_{1,3} & \hat{m}^{(i)}_{2,3} & 0
	\end{bmatrix}. \nonumber
	\end{align}

	For $0\leq i\leq \nu-1=1$, the $i^\text{th}$ MBR local code is obtained by evaluating the vector of polynomials $[\hat{m}^{(i)}_0(x)\ \hat{m}^{(i)}_1(x)\ \hat{m}^{(i)}_2(x)\ \hat{m}^{(i)}_3(x)]$ at the $n_\ell=6$ evaluation points given by $\mathcal{A}^{(i)}\triangleq\gamma^i\{1,\gamma^2,\gamma^4,\gamma^6,\gamma^{8},\gamma^{10}\}$. Here $\hat{m}^{(i)}_j(x)\triangleq \sum_{t=0}^{d-1}\hat{m}^{(i)}_{j,t}x^t$. In order to stress up on the symmetric nature of $\hat{\mathbf{M}}^{(i)}$, we relabel $\hat{m}^{(i)}_{x,y}$ as $\hat{m}^{(i)}_{y,x}$, whenever $x>y$. Thus we have:
	$\hat{m}^{(i)}_0(x)\triangleq\hat{m}^{(i)}_{0,0}+\hat{m}^{(i)}_{0,1}x+\hat{m}^{(i)}_{0,2}x^2+\hat{m}^{(i)}_{0,3}x^3,\ldots,\hat{m}^{(i)}_3(x)\triangleq\hat{m}^{(i)}_{0,3}+\hat{m}^{(i)}_{1,3}x+\hat{m}^{(i)}_{2,3}x^2$. 
	
	Consider the column, $j=(d-1)=3$, of the $\nu=2$ message matrices. We have the $\nu=2$ polynomials; $\hat{m}^{(0)}_3(x)$ and $\hat{m}^{(1)}_3(x)$. Take $\hat{m}^{(i)}_3(x)$ to be the polynomial $\tilde{m}^{(i)}(x)$ appearing in \eqref{eq:poly_lifting}, for $0\leq i\leq \nu-1=1$. We shall perform polynomial lifting to arrive at the polynomial $\hat{M}_3(x)$, which has the property (similar to that of $\tilde{M}(x)$ stated in \eqref{eq:poly_lifting}): $\hat{M}_3(x) \mod f^{(i)}(x) =  \hat{m}_3^{(i)}(x)$, for all $0\leq i\leq \nu-1=1$. Note that the vector space $\mathcal{\hat{M}}_3$ of all possible $\hat{M}_3(x)$'s, has a dimension of $r\nu=3*2=6$, as we don't assume any dependencies across the $j=3$ column of $\hat{\mathbf{M}}^{(0)}$ and $\hat{\mathbf{M}}^{(1)}$, to start with. Let $\hat{M}_3(x)\triangleq\sum_{t=0}^{n-1}\hat{M}_{3,t}x^t$. For $0\leq t\leq (n-1)$, $0\leq a\leq \nu-1$ and $0\leq b\leq n_\ell-1$,  let $t=an_\ell+b$. We have:
	\begin{equation}\label{eq:crt_map_MBR}
	\hat{M}_{3,t}  =
	\begin{cases} 
	\sum_{s=0}^{\nu-1}e^{(s)}_{a}\hat{m}^{(s)}_{3,b} & b \leq (r-1), \\
	0 & \text{otherwise.}
	\end{cases}
	\end{equation}
	Here \eqref{eq:crt_map_MBR} is just a restatement of \eqref{eq:crt_map}. Note that $\mathcal{\hat{M}}_3$ is the space spanned by the set of monomials $\{x^t:t\in \{0,1,2,6,7,8\}\}$. Let $\mathcal{M}_3$ be the subspace of  $\mathcal{\hat{M}}_3$ spanned by  $\{x^t:t\in \{0,1,2,6\}\}$, after removing the two largest degree terms (in Example \ref{eg:TB_code_example}, we obtained $\mathcal{M}$ from $\tilde{\mathcal{M}}$ in a similar manner). By CRT, this essentially means introducing two dependencies across the coefficients of polynomials  $\hat{m}_3^{(0)}(x)$ and $\hat{m}_3^{(1)}(x)$. These dependencies introduced for $t=8$ and $t=7$, are explicitly given by:
	
	\begin{equation}\label{eq:mbr_code_dependencies}
	\sum_{s=0}^{\nu-1}e^{(s)}_{1}\hat{m}^{(s)}_{3,b}=0,\ \text{for}\  b\in\{1,2\}
	\end{equation} 
	
	After relabeling $\hat{m}^{(s)}_{3,b}$ as $\hat{m}^{(s)}_{b,3}$, we have:
	\begin{equation*}\label{eq:mbr_dependency}
	e^{(0)}_{1}\hat{m}^{(0)}_{2,3}+e^{(1)}_{1}\hat{m}^{(1)}_{2,3}=0,\ 
	e^{(0)}_{1}\hat{m}^{(0)}_{1,3}+e^{(1)}_{1}\hat{m}^{(1)}_{1,3}=0.
	\end{equation*} 
	
	Now consider the column, $j=2$. Similar to the case of $j=3$, consider the $\nu=2$ polynomials $\hat{m}^{(0)}_2(x)$ and $\hat{m}^{(1)}_2(x)$. After performing the polynomial lifting, we arrive at the polynomial $\hat{M}_2(x)$, which belongs to the space spanned by $\{x^t:t\in \{0,1,2,3,6,7,8,9\}\}$. We then introduce three dependencies among the polynomials $\hat{m}^{(0)}_2(x)$ and $\hat{m}^{(1)}_2(x)$ by considering the subspace
	$\mathcal{{M}}_2$ of $\mathcal{\hat{M}}_2$ spanned by $\{x^t:t\in \{0,1,2,3,6\}\}$.
	
	We perform an identical operation for $j=1$ as well, with three dependencies. In Table \ref{tab:dependencies_list}, we summarize the three cases $j=3,2,1$. There are five unique dependencies introduced, and hence the dimension of the space of all possible $\{[\hat{m}^{(i)}_0(x),\hat{m}^{(i)}_1(x), \hat{m}^{(i)}_2(x), \hat{m}^{(i)}_3(x)]\}_{i=0}^{\nu-1=1}$ will be $\nu K_\ell-5=13=K$. Now, for $0\leq i\leq \nu-1=1$, $i^\text{th}$ local codewords are produced using $\hat{\mathbf{M}}^{(i)}$ along with the evaluation points $\mathcal{A}^{(i)}$, as in Example \ref{eg:MBR_poly_eval}. Using Remark \ref{remark:MDS_remark}, one can infer that even after introducing dependencies, the vector space of all possible $\hat{\mathbf{M}}^{(i)}$'s, for a fixed $i\in[0,1]$, will still have the dimension $K_\ell=9$. Hence the code $\mathcal{C}$ thus formed, is a code with MBR local regeneration having all the desired parameters. Figure \ref{fig:mbr_locality_eg} gives an illustration of this example code.  
	
	\begin{table}[ht]
		\centering
		\begin{tabular}{ |c|c|c|c|c|c| } 
			
			\hline
			Column&Number of&Dependencies\\ 
			& dependencies & \\	
			\hline \hline
			$3$ & $2$ & $e^{(0)}_{1}\hat{m}^{(0)}_{2,3}+e^{(1)}_{1}\hat{m}^{(1)}_{2,3}=0$,\\ 
			& & $e^{(0)}_{1}\hat{m}^{(0)}_{1,3}+e^{(1)}_{1}\hat{m}^{(1)}_{1,3}=0$\\
			\hline
			$2$ & $3$ & $e^{(0)}_{1}\hat{m}^{(0)}_{2,3}+e^{(1)}_{1}\hat{m}^{(1)}_{2,3}=0$,\\
			&  & $e^{(0)}_{1}\hat{m}^{(0)}_{2,2}+e^{(1)}_{1}\hat{m}^{(1)}_{2,2}=0$,\\
			&  & $e^{(0)}_{1}\hat{m}^{(0)}_{1,2}+e^{(1)}_{1}\hat{m}^{(1)}_{1,2}=0$ \\
			\hline
			$1$ & $3$ & $e^{(0)}_{1}\hat{m}^{(0)}_{1,3}+e^{(1)}_{1}\hat{m}^{(1)}_{1,3}=0$,\\
			&  & $e^{(0)}_{1}\hat{m}^{(0)}_{1,2}+e^{(1)}_{1}\hat{m}^{(1)}_{1,2}=0$,\\
			&  & $e^{(0)}_{1}\hat{m}^{(0)}_{1,1}+e^{(1)}_{1}\hat{m}^{(1)}_{1,1}=0$ \\
			\hline
		\end{tabular}
		\caption{A summary of dependencies introduced across the columns $3$, $2$ and $1$, of message matrices $\hat{\mathbf{M}}^{(0)}$ and $\hat{\mathbf{M}}^{(1)}$.}
		\label{tab:dependencies_list}
	\end{table}
\end{example}

As for the minimum-distance, the scalar code of length $12$, obtained by restricting $\mathcal{C}$ to any column $j\in\{3,2,1,0\}$ from each node, will be a TB code. Restricted to column $3$, using Remark \ref{remark:eval_match}, the TB codeword obtained will be $n$ evaluations of a polynomial lying in the span of $\{x^t:t\in \{0,1,2,6\}\}$. Similarly, columns $2$ and $1$ yield scalar codes, which are evaluations of polynomials lying in the span of $\{x^t:t\in \{0,1,2,3,6\}\}$. As the degree of these polynomials is at most $6$, minimum-distance restricted to these columns will be at least $(n-6)=6$. 

If all the scalar codewords obtained from columns $3$, $2$ and $1$ are zero-codewords, it essentially implies $\hat{m}^{(i)}_j(x)$'s are all zero-polynomials for $j\in\{3,2,1\}$ and $i\in\{0,1\}$. As $\hat{\mathbf{M}}^{(i)}$'s are all symmetric matrices, for these cases, $\hat{m}^{(i)}_{x,y}=0$ $\forall x,y,i: 1\leq x,y\leq 3,0\leq i\leq 1$. Hence, only $\hat{m}^{(i)}_{0,0}$'s can possibly be non-zero. Thus, restricted to column $0$, the scalar codeword will be $n$ evaluations of a polynomial lying in the span of $\{x^t:t\in \{0,6\}\}$. Hence even for column $0$, minimum-distance will be at least $6$. This essentially proves the minimum-distance optimality of the code.  

\begin{figure}[ht]
	\centering
	\captionsetup{justification=centering}
	\includegraphics[scale=0.45]{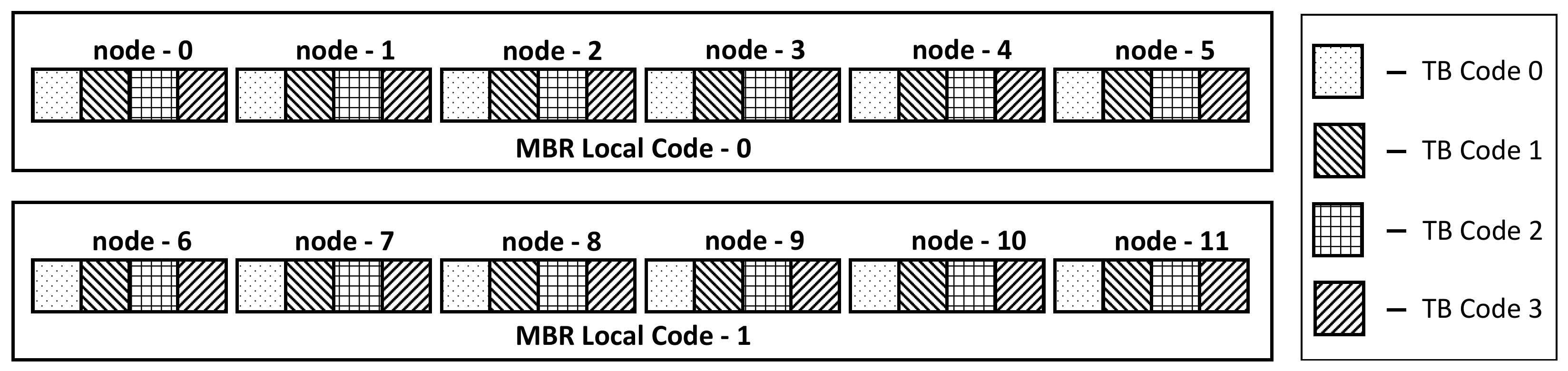}
	\caption{An illustration of the code in Example \ref{eg:mbr_locality}.}
	\label{fig:mbr_locality_eg}
\end{figure}

\section{Codes with MSR Locality}\label{sec:MSR_local_regen}

From Remark \ref{remark:MDS_remark}, we know that each local code in a TB code is an MDS code. In order to construct a code with MSR local regeneration, we initially stack $\alpha=(n_\ell-r)^{\frac{n_\ell}{n_\ell-r}}$ independent layers of codewords from an $[n,k,d_\text{TB}]$ TB code with $(r,\delta)$ all-symbol locality. We then perform the PCT independently, for each local code. This essentially results in a code $\mathcal{C}$ with MSR local regeneration. The local MSR code will have the parameters $((n_\ell,r,d),(\alpha,\beta),K_\ell)$, with $d=n_\ell-1$, $(n_\ell-r)\mid r$. Let $d_\text{TB}$ denote the (optimal) minimum-distance of the underlying TB code.

\begin{thm}
		$\mathcal{C}$ has optimal minimum-distance when $d_\text{TB}\leq 2\delta$.
\end{thm}		
\begin{proof}
	First we show that $d_{\text{min}}(\mathcal{C})\geq d_\text{TB}\leq 2\delta$. Assume to the contrary that $d_{\text{min}}(\mathcal{C})<d_\text{TB}\leq 2\delta$. Consider the vector codeword of $\mathcal{C}$ with hamming weight $d_{\text{min}}(\mathcal{C})<2\delta$. As each local MSR code has a minimum-distance of $\delta$, all the vector code-symbols having non-zero weights must be restricted within a local MSR code. From Corollary \ref{cor:minimum_distance_layer0}, there exists an underlying TB codeword with one local codeword having hamming weight $\leq d_{\text{min}}(\mathcal{C})<d_\text{TB}$ and all other local codewords as zeros, which is a contradiction. As $K=k\alpha$, where $k$ is the dimension of the TB code, \eqref{eq:msr_locality_bound} reduces to $d_{\text{min}}(\mathcal{C})\leq d_\text{TB}$. This completes the proof.
\end{proof}

\bibliographystyle{IEEEtran}
\bibliography{loc_and_dec_v2}

\end{document}